\documentclass{article}

\usepackage{amsthm,amsmath,amssymb,enumerate,hyperref, mathtools}
\usepackage[normalem]{ulem}
\usepackage[backend=biber]{biblatex}
\newtheorem{theorem}{Theorem}
\newtheorem{proposition}[theorem]{Proposition}
\newtheorem{definition}[theorem]{Definition}
\newtheorem*{definition*}{Definition}
\newtheorem{corollary}[theorem]{Corollary}

\newtheorem{remark}[theorem]{Remark}
\newtheorem{observation}[theorem]{Observation}

\newtheorem{lemma}[theorem]{Lemma}

\usepackage{amsmath,amssymb,mathtools}
\newcommand{\R}{\mathbb{R}}
\newcommand{\N}{\mathbb{N}}
\newcommand{\Q}{\mathbb{Q}}
\newcommand{\Z}{\mathbb{Z}}
\newcommand{\Dim}{\textup{Dim}}
\newcommand{\CE}{\text{CE}}

\title{Finite-State Dimension and The Davenport-Erd\H{o}s Theorem}
\author{Joe Clanin\footnote{Department of Computer Science, Iowa State University, \href{mailto:jsc@iastate.edu}{jsc@iastate.edu}. This author's research was supported in part by National Science Foundation research grant 1900716.} \and Matthew Rayman\footnote{Department of Computer Science, Iowa State University, \href{mailto:marayman@iastate.edu}{marayman@iastate.edu}. This author's research was supported in part by National Science Foundation research grant 1900716.}}
\date{}

\begin{document}

\maketitle

\begin{abstract}
A 1952 result of Davenport and Erd\H{o}s states that if $p$ is an integer-valued polynomial, then the real number $0.p(1)p(2)p(3)\dots$ is Borel normal in base ten. A later result of Nakai and Shiokawa extends this result to polynomials with arbitrary real coefficients and all bases $b\geq 2$. It is well-known that finite-state dimension, a finite-state effectivization of the classical Hausdorff dimension, characterizes the Borel normal sequences as precisely those sequences of finite-state dimension 1. For an infinite set of natural numbers, and a base $b\geq 2$, the base $b$ Copeland-Erd\H{o}s sequence of $A$, $CE_b(A)$, is the infinite sequence obtained by concatenating the base $b$ expressions of the numbers in $A$ in increasing order. In this work we investigate the possible relationships between the finite-state dimensions of $\CE_b(A)$ and $\CE_b(p(A))$ where $p$ is a polynomial. We show that, if the polynomial is permitted to have arbitrary real coefficients, then for any $s,s^\prime$ in the unit interval, there is a set $A$ of natural numbers and a linear polynomial $p$ so that the finite-state dimensions of $\CE_b(A)$ and $\CE_b(p(A))$ are $s$ and $s^\prime$ respectively. The corresponding result for strong finite-state dimension is also shown. We demonstrate that linear polynomials with rational coefficients do not change the finite-state dimension of any Copeland-Erd\H{o}s sequence, but there exist polynomials with rational coefficients of every larger integer degree that change the finite-state dimension of some sequence. We also prove the surprising fact that there exist sets $A$ and integer-valued monomials $p$ such that $CE_b(A)$ is normal, but $CE_b(p(A))$ has finite-state dimension strictly less than one.
\end{abstract}

\section{Introduction}

A real number $x$ is said to be Borel normal in base $b$ if the base-$b$ expansion of $x$ contains, for every positive integer $n$, each string of length $n$ over the $b$-ary alphabet with limiting density $b^{-n}$. While it has long been known that almost all real numbers are normal \cite{emile1909probabilites}, and it is conjectured \cite{bugeaud2012distribution} that every irrational algebraic number is absolutely normal (i.e. normal to every positive integer base), all known explicit examples of normal numbers have arisen through artificial construction \cite{becher2018normal}. The first such example, the base ten Champernowne constant, \[0.12345678910\dots\] obtained by concatenating the base ten representations of the positive integers was shown to be normal in 1933 \cite{champernowne1933construction}. A later work by Besicovitch \cite{besicovitch1935asymptotic} demonstrated the same result for the concatenations the square numbers in base ten. Copeland and Erd\H{o}s \cite{copeland1946note} further demonstrated this for the primes and derived a simple sufficient condition in terms of natural density for the concatenation, in increasing order, of an infinite set of natural numbers to be a normal sequence in any given positive integer base. 

The present work has its origins in a 1952 construction of Davenport and Erd\H{o}s. In \cite{davenport1952note}, they proved that if $p(x)$ is a non-constant polynomial taking only positive integer values on the positive integers, then the real number \[0.p(1)p(2)p(3)p(4)\dots\] is normal in base ten. Their result was extended by Nakai and Shiokawa \cite{nakai1990classII,nakai1990class,shiokawa1974remark} who showed that for any function of the form \[f(x)=\alpha_nx^\beta_n+\alpha_{n-1}x^\beta_{n-1}+\dots+\alpha_1x^\beta_1\] where $\beta_n>\beta_{n-1}>\dots>\beta_1\geq 0$ and $f(x)>0$ for $x>0$, the concatenation \[0.[f(1)]_b[f(2)]_b[f(3)]_b[f(4)]_b\dots\] where $[x]_b$ represents the integer part of $x$ in base $b$ is normal in base $b$. 

Finite-state dimension, first introduced by Dai, Lathrop, Lutz and Mayordomo \cite{dai2004finite} is an effectivized version of the classical Hausdorff dimension. The finite-state dimension of an infinite sequence $S$ (or real number), $\dim_{FS}(S)$, quantifies the lower asymptotic density of finite-state information contained in $S$. The dual notion of strong finite-state dimension, $\Dim_{FS}(S)$, is a finite-state effectivization of the classical packing dimension and correspondingly quantifies the upper asymptotic density of finite-state information in $S$. It is well established \cite{bourke2005entropy,schnorr1972endliche} that a sequence $S$ is normal if and only if $\dim_{FS}(S)=1$. In general, however, it is the case that $0\leq\dim_{FS}(S)\leq\Dim_{FS}(S)\leq 1$ and these two quantities may differ.

A direction for research first proposed by Jack Lutz in 2005 involves investigating the claim: Every theorem about Borel normality is the dimension-1 special case of a more interesting theorem about finite-state dimension. This claim, now known as the ``Normality/Dimension Thesis" has inspired a number of subsequent works. For example, the aforementioned results of Copeland and Erd\H{o}s \cite{copeland1946note} were extended in \cite{gu2005dimensions}, and the 1949 theorem of Wall \cite{wall1949normal} asserting the normality of $q+x$ and $qx$ for normal $x\in\R$ and $q\in\Q\setminus\{0\}$ was shown to hold for finite-state dimension in \cite{doty2007finite}. In this work, we seek to address the Davenport-Erd\H{o}s theorem and its subsequent generalization by Nakai and Shiokawa through the lens of this thesis. Existing literature concerning constructions like that of Davenport and Erd\H{o}s has sought to produce a normal number (and therefore a dimension-1 sequence) by applying a function to $\N$ or the set of primes. As both of these sets themselves give normal numbers when concatenated, we seek in the present work to investigate the finite-state dimension of the sequence $p(n_1)p(n_2)p(n_3)\dots$ when the sequence $n_1n_2n_3\dots$ is permitted to have a finite-state dimension other than 1, and $p$ is a polynomial.

The following section introduces our notational conventions and definitions. In Section 3 we address the case of polynomials with arbitrary real coefficients, in Section 4 we consider the case of polynomials with rational coefficients, and in section 5 we describe the behavior of a related quantity, zeta-dimension, under polynomial transformations.

\section{Preliminaries}

Let $\Sigma$ denote a finite alphabet, and for $b\geq 2$ let $\Sigma_b$ denote the $b$-ary alphabet $\{0,\dots,b-1\}$. The set of all finite strings over $\Sigma$ is denoted $\Sigma^*$. The set of strings in $\Sigma^*$ of length $l$ is denoted $\Sigma^l$, and $\Sigma^\infty$ denotes the set of infinite sequences over $\Sigma$. For any $w\in \Sigma^*\cup\Sigma^\infty$, $w[i\dots j]$ is the substring of $w$ beginning at the index $i$ and terminating at the index $j$.

Following the notation in \cite{gu2005dimensions}, for $A\subseteq \N$ and $b\geq 2$, the base $b$ Copeland-Erd\H{o}s sequence $\CE_b(A)$ is the infinite $b-$ary sequence obtained by concatenating the $b-$ary expressions of the numbers in $A$ in increasing order of their magnitude. For example, \[
\CE_2(\N)=011011100101\dots\text{ and }\CE_{10}(\text{PRIMES})=235711131719\dots.\] For any real number $x\in\R$, denote the base-$b$ expression of $x$ by $\sigma_b(x)$. For $f:\N\to\R$ and $A\subseteq \N$, define \[\CE_b(f(A))=\sigma_b\left([ f(n_1)]\right)\sigma_b\left([ f(n_2)]\right)\dots\] where $[\hspace{0.5em}]$ denotes the integer part and $n_1<n_2<\dots$ are the elements of $A$.

\begin{definition}(\cite{bourke2005entropy}) Let $w,s\in \Sigma^*$.
\begin{itemize}
    \item The number of sliding block occurences of $w$ in $s$ is the quantity \[N(w,s)=\left|\{i | s[i...i+|w|-1]=w\}\right|.\]
    \item The sliding count probability of $w$ in $s$ is \[P(w,s)=\frac{1}{|s|-|w|+1}N(w,s).\]
    \item For $l\in\N$ the $l^{th}$ sliding block entropy of $x\in\Sigma^*$ is
\[H_l(x)=\frac{1}{l}\sum_{w\in\Sigma^l}P(w,x)\log(P(w,x)^{-1}).\]
\end{itemize}
\end{definition}

Finite-state dimension admits many equivalent characterizations including finite-state gamblers and finite-state compressors \cite{dai2004finite}, finite-state log-loss predictors \cite{hitchcock2003fractal}, and automatic Kolmogorov complexity \cite{kozachinskiy2021automatic}. In this work, we utilize a characterization in terms of sliding block entropy. Disjoint block entropy was shown to be a characterization by Bourke, Hitchcock, and Vinodchandran \cite{bourke2005entropy} and Kozachinskiy and Shen \cite{kozachinskiy2019two} showed that the following sliding block version is equivalent.

\begin{definition}\label{fsd_definition}(\cite{kozachinskiy2019two},\cite{dai2004finite})
The finite-state dimension and finite-state strong dimension of a sequence $S\in\Sigma^\infty$ are given by
\[\dim_{FS}(S)=\inf_{l}\liminf_{n\to\infty}H_l(S[0\dots n])\hspace{1em}\text{ and }\hspace{1em}\Dim_{FS}(S)=\inf_{l}\limsup_{n\to\infty}H_l(S[0\dots n]).\]
\end{definition}

A sequence $S\in \Sigma_b^\infty$ (or equivalently, a real number) is said to be Borel normal in base $b$ if, for all $w\in \Sigma_b^*$,
\[\lim_{n\to\infty}\frac{N(w,S[0\dots n])}{n}=\frac{1}{b^{|w|}}.\]

It is well-known \cite{bourke2005entropy,schnorr1972endliche} that the set of normal sequences is precisely the set of sequences of finite-state dimension 1, and that sequences of all possible dimensions and strong dimensions exist. In particular, it has been shown \cite{gu2005dimensions} that for any $s,t\in[0,1]$ with $s<t$, there exists $A\subseteq \N$ such that \[\dim_{FS}(A)=s\text{ and }\Dim_{FS}(A)=t.\]

We will make repeated use of the following result, for which we provide a short proof.

\begin{proposition}\label{density_zero_lemma}
    Let $S,T\in \Sigma^\infty$. If $T$ can be obtained by inserting or deleting symbol in $S$ at a set of indices of upper asymptotic density zero, then $\dim_{FS}(S)=\dim_{FS}(T)$ and $\Dim_{FS}(S)=\Dim_{FS}(T)$.
\end{proposition}
\begin{proof}
    Let $l\in\mathbb{N}$. For any $w\in \Sigma^*$,  as $n\to\infty$ we have that $|P(w,S[0\dots n])-P(w,T[0\dots n])|\to 0$ and thus also $|H_l(S[0\dots n])-H_l(T[0\dots n])|\to 0$. Therefore $\dim_{FS}(S)=\dim_{FS}(T)$ and $\Dim_{FS}(S)=\Dim_{FS}(T)$.
\end{proof}

\section{Polynomials with Arbitrary Real Coefficients}

In this section, we investigate the finite-state dimensions of $\CE_b(A)$ and $\CE_b(p(A))$ where $p$ is a polynomial with arbitrary real coefficients. We will use the following result based on the concavity of $H_l$ block entropies. Let $P_1$ and $P_2$ be two probability distributions over $\Sigma^l$ and $0\leq \lambda \leq 1$. Then
\[H_l(\lambda P_1 + (1-\lambda)P_2 )\geq  \lambda H_l(P_1)+(1-\lambda)H_l(P_2)\]
where $H_l(P)$ is defined in the natural way for a probability distribution instead of a string.
The proof is analogous to the proof for Shannon entropy. In particular, we have the following.

\begin{proposition}\label{concavity}
    Let $u,v\in \Sigma^*$ and $w=uv$. Then
\[H_l(w)\geq  \frac{|u|}{|w|}H_l(u)+\frac{|v|}{|w|}H_l(v).\]
\end{proposition}
\begin{proof}
    Let $\lambda=\frac{|u|}{|w|}$ and $P_1, P_2$ be the probability distributions corresponding to the frequencies in the string $u$ and $v$ respectively.
\end{proof}
In our constructions we will concatenate prefixes of a given sequence. The following result allows us to do this without changing the dimension for sufficiently fast growing prefixes.
\begin{proposition}
Let $S\in \Sigma^\infty$ with $\dim_{FS}(S)=s$. Then there exists $n_1<n_2<n_3<...\in \mathbb{N}$ such that $T=S[0...n_1]S[0...n_2]S[0...n_3]...$ has $\dim_{FS}(T)=s$.
\end{proposition}\label{prefix_lemma}
\begin{proof}
    Let $(a_n)_{n\in \mathbb{N}}=(1,1,2,1,2,3,...)$. Define $n_i$ inductively as follows at stage $i$:
    \begin{itemize}
        \item Let $S_{i-1}=S[0...n_1]...S[0...n_{i-1}]$. 
        \item For $l<i$ let $s_l$ be such that $H_l(S[0...k])\geq \liminf_{k\rightarrow \infty} H_l(S[0...k])-2^{-i}$ for all $k>s_l$. Let $t_l$ be such that ${H_l(S_{i-1}S[0...k])\geq \liminf_{k\rightarrow \infty} H_l(S[0...k])-2^{-i-1}}$ for all $k>t_l$. 
        \item Let $k_l=(i)(t_l)(s_l)$.
        \item Pick $n_i>\max\{n_{i-1},k_l\vert l\leq i\}$ with \[{H_{a_i}(S_{i-1}S[0...n_i])\leq \liminf_{k\rightarrow \infty} H_{a_i}(S[0...k])+2^{-i}}.\]
    \end{itemize}
    \par
    By the last criterion, for each $l$ there are infinitely many $k$ such that $H_l(T[0...k])\leq \liminf_{k\rightarrow \infty} H_{l}(S[0...k])+2^{-i}$ for every $i$ so 
    \[\liminf_{k\rightarrow \infty} H_{l}(T[0...k])\leq \liminf_{k\rightarrow \infty} H_{l}(S[0...k]). \]
    \par 
    By Proposition \ref{concavity}, for each $i$ 
    \[H_l(S_{i-1}S[0...k])\geq \frac{|S_{i-1}|}{|S_{i-1}|+k}H_l(S_{i-1})+\frac{k}{|S_{i-1}|+k}H_l(S[0...k]).\]
    For $k<s_l$  we have
    \[H_l(S_{i-1}(S[0...k])\geq \frac{i-1}{i}(\liminf_{k\rightarrow \infty} H_l(S[0...k])-2^{-i})\]
    and for $k>s_l$,
    \[H_l(S_{i-1}(S[0...k])\geq \min\{H_l(S_{i-1}), H_l(S[0...k])\}\geq \liminf_{k\rightarrow \infty} H_l(S[0...k])-2^{-i}).\]
    So for all $k$,
    \[H_l(S_{i-1}(S[0...k])\geq \frac{i-1}{i}(\liminf_{k\rightarrow \infty} H_l(S[0...k])-2^{-i}).\]
    Thus, 
    \[\liminf_{k\rightarrow \infty} H_{l}(T[0...k])\geq \liminf_{k\rightarrow \infty} H_{l}(S[0...k]) \]
    and hence,
    \[\liminf_{k\rightarrow \infty} H_{l}(T[0...k])= \liminf_{k\rightarrow \infty} H_{l}(S[0...k]). \]
\end{proof}

The above result also holds for strong finite-state dimension, however the proof requires the use of a different characterization of finite-state dimension than what is given by definition~\ref{fsd_definition}. The alternative characterization is given in terms of betting on sequences using a finite-state machine called a {\it finite-state gambler}.  We refer the reader to \cite{athreya2007effective} and \cite{dai2004finite} for the relevant definitions regarding finite-state gamblers.

\begin{proposition}\label{strongprefixlemma}
    Let $S\in \Sigma^\infty$ with $\Dim_{FS}(S)=t$. Then there exists $n_1<n_2<n_3<...\in \mathbb{N}$ such that $T=S[0...n_1]S[0...n_2]S[0...n_3]...$ has $\Dim_{FS}(T)=t$.
\end{proposition}
\begin{proof}
    Let $(a_n)_{n\in \mathbb{N}}=(1,1,2,1,2,3,...)$ and let $(s_n)_{n=0}^\infty$ be a sequence of rationals approaching $t$ from above. Since each $t_n<\Dim_{FS}(S)$ there is a finite-state gambler $d_n$ that $t_n$-strongly succeeds on $s$. We inductively define $n_i$ as follows.
    \begin{itemize}
        \item For $j\leq i$ let $r_i\in \mathbb{N}$ be least length such that the state of $d_j$ after processing $S[0\dots n_{i-1}]$ is the same as after processing $S[0\dots r_i]$. Define $m_j$ such that $d(S[0\dots r_i])\leq 2d(S[0\dots m_j])$ for all $m\geq m_j$.
        \item Let $N_1= \max\{m_j \mid j\leq i\}$.
        \item Let $N_2$ be such that 
        \[{H_{a_i}(S[0\dots n_1]\dots S[0\dots n_{i-1}]S[0\dots N_2])\geq \limsup_{k\rightarrow \infty} H_{a_i}(S[0...k])-2^{-i}}.\]
        \item Choose $n_i=\max\{N_1,N_2,n_{i-1}\}.$
    \end{itemize}
    Note that the values of $N_2$ above ensure that $\Dim_{FS}(T)\geq \Dim_{FS}(S)=t$. Therefore it suffices to prove that $\Dim_{FS}(T)\leq t$. To do this we will show that $\Dim_{FS}(T)\leq s_n$ for every $s_n$. Note that there are finitely many states and for each the corresponding $r$ defined in item 1 above has constant length. Hence for the sequence
    \[T'=S[0\dots n_1]S[r_1\dots n_2]\dots\]
    we get that $\Dim_{FS}(T')=\Dim_{FS}(T)$ by Proposition \ref{density_zero_lemma}. 
    Therefore the proof follows if we show that the gambler $d_n$ $t_n$-strongly succeeds on $T'$. 

    To see this let $q=d_n(S[0\dots n_1]\dots S[r_{n-2}\dots n_{n-1}])>0$. The choice of $m_n$ guarantees that $d_n$ will have at least twice its value after processing every added prefix of $S$ after this point. After $k$ prefixes have been added we then have 
    \[d(S[0\dots n_1]S[r_{n+k-1}\dots n_{n+k}]S[r_{n+k} \dots m])\geq q2^kd(S[0\dots m])\]
     for all $m\geq r_{n+k}$ and thus $d_n$ $t_n$-strongly succeeds on $T'$.
\end{proof}

\begin{observation}
    Given finitely many sequences $S_1,S_2,\dots S_n$, a single sequence of naturals $n_1<n_2<n_3\dots \in \mathbb{N}$ exists such that Lemma \ref{prefix_lemma} holds for each $S_i$.
\end{observation}
\begin{proof}
    Use the same construction as the proof of Proposition \ref{prefix_lemma} and choose $n_i$ to be the max of those obtained for each sequence.
\end{proof}

It is natural to wonder what happens for concatenations of prefixes which are not of fast-growing length. The following result gives an answer to this by looking at concatenating every single prefix. In particular, note that for any normal sequence $S$ the sequence $T$ defined below is normal.

\begin{proposition}\label{allprefixesproposition}\footnote{The proof of this proposition is located in the appendix}
    Let $S\in \Sigma^\infty$ and $T=S[0]S[0...1]S[0...2]...$ Then $\dim_{FS}(S)\leq \dim_{FS}(T)$.
\end{proposition}

\begin{remark}
    It is straightforward to construct, for each $b\geq 2$, an $S\in\Sigma_b^\infty$ and $n_1<n_2<n_3<\dots\in\N$ so that \[\dim_{FS}(S)\lneqq\dim_{FS}(S[0\dots n_0]S[0\dots n_1]S[0\dots n_2]\dots)\] by dilution of a normal sequence (\cite{dai2004finite}, construction 6.4) and careful choice of prefix lengths.
\end{remark}

\begin{proposition}\label{linear_polynomial_result}
    For any $b\geq 2$ and $s,s^\prime,t,t^\prime \in [0,1]$ with $s\leq t$ and $s^\prime \leq  t^\prime$ there exists an infinite set $A\subseteq \N$ and a linear polynomial $p(x)$ so that
    \[\dim_{FS}(\CE_b(A))=s,\hspace{0.5em}\Dim_{FS}(\CE_b(A))=t\]
    and \[\dim_{FS}(\CE_b(p(A))=s^\prime,\hspace{0.5em}\dim_{FS}(\CE_b(p(A))=t^\prime\]
    \end{proposition}
\begin{proof}
    Let $\alpha\in\R$ be such that $\dim_{FS}(\sigma_b(\alpha))=s$ and $\Dim_{FS}(\sigma_b(\alpha))=t$. Take $c\in\R$ so that $c\alpha$ has $\dim_{FS}(\sigma_b(c\alpha))=s^\prime$ and $\Dim_{FS}(\sigma_b(c\alpha))=t^\prime$. By Propositions~\ref{prefix_lemma},\ref{strongprefixlemma}, we can choose $n_1<n_2<\dots\in \N$ sufficiently fast increasing so that $p(x)=cx$ and $A=\{\sigma_b(\alpha)[0\dots n_i] | i\in \N\}$ yield the result noting that one can choose the max of the two corresponding values of $n_i$ and construct the sequence simultaneously meeting all conditions in the proof of the propositions. As the set of indices at which the base $b$ representations of each product $cn$ for $n\in A$ disagree with the corresponding prefix of $c\alpha$ form a set of asymptotic density zero, we have $\dim_{FS}(\CE_b(p(A))=\dim_{FS}(c\alpha)$ and $\Dim_{FS}(\CE_b(p(A))=\Dim_{FS}(c\alpha)$.
\end{proof}

The above result relies critically on our ability to encode information into the irrational coefficient $c$ in the linear polynomial. In the following section, we investigate the behavior of polynomials with rational coefficients.

\section{Polynomials with Rational Coefficients}

\begin{proposition}\label{zero-blocks}
    Let $S$ be a normal sequence and $S^0$ be a sequence obtained by inserting strings in $\{0\}^*$ such that the number of strings inserted to $S$ has upper asymptotic density zero relative to the number of digits in the prefix. Let $z(n)$ be the number of zeroes added into a prefix of $S$ when viewed as a prefix of $S^0$ of length $n$. Then 
    \[\dim_{FS}(S^0)=1-\limsup_{n\to\infty}\frac{z(n)}{n}\]
    and 
    \[\Dim_{FS}(S^0)=1-\liminf_{n\to\infty}\frac{z(n)}{n}\]
\end{proposition}
\begin{proof}
    Let $l\in \mathbb{N}$. Then the number of sliding blocks over the original parts of $S$ that are impacted by the additions of 0's has asymptotic density zero. Hence, in the limit the probability distribution of length $l$ strings is equivalent to the weighted sum of the probability distribution $P_{0,l}$ of $S^0$ on the sliding blocks over only the added zero bits and the probability distribution $P_{S,l}$ on the original bits present in $S$. In particular, letting $P_{S^0,l,n}$ be the probability distribution of length $l$ strings on a prefix of $S^0$ of length $n$ we have 
\[P_{S^0,l,n}(0^l)=\frac{z(n)}{n}+\frac{n-z(n)}{n}P_{S,l}(0^l)\]
and 
\[P_{S^0,l,n}(w)=\frac{n-z(n)}{n}P_{S,l}(w)\]
for all other $w\in \Sigma^l$ as $n$ goes to infinity.
Since $S$ is normal we have
\[P_{S^0,l,n}(0^l)=\frac{z(n)}{n}+\frac{n-z(n)}{n}\frac{1}{b^l}\]
and
\[P_{S^0,l,n}(w)=\frac{n-z(n)}{n}\frac{1}{b^l}.\]
as $n$ goes to infinity. Therefore, 
\begin{align}
\liminf_{n\to \infty} H_l(S^0[0\dots n]) &=\liminf_{n\to \infty} -\frac{1}{l}\biggl[\left(\frac{z(n)}{n}+\frac{n-z(n)}{n}\frac{1}{b^l}\right)\log\left(\frac{z(n)}{n}+\frac{n-z(n)}{n}\frac{1}{b^l}\right)\nonumber \\
  &+(b^l-1)\left(\frac{n-z(n)}{n}\frac{1}{b^l}\right)\log\left(\frac{n-z(n)}{n}\frac{1}{b^l}\right)\biggr]\nonumber
\end{align}
It is routine to verify that this is non-increasing with $l$ and converges to \\${1-\limsup_{n\to\infty}{\frac{z(n)}{n}}}$ as $l$ goes to infinity. The $\limsup$ result is analogous. 
\end{proof}
Note that the $l$-block entropy in the weighted distribution is minimized when the probability distribution of the inserted strings has $l$-block entropy zero (e.g. contains only one string of length $l$). Therefore we get the following as a corollary.
\begin{corollary}\label{garbage-blocks}
    Let $S$ be a normal sequence and $S'$ be a sequence obtained by inserting arbitrary strings such that the number of strings inserted to $S$ has upper asymptotic density zero relative to the number of digits in the prefix. Let $z'(n)$ be the number of digits added into a prefix of $S$ when viewed as a prefix of $S'$ of length $n$. Then 
    \[\dim_{FS}(S')\geq1-\limsup_{n\to\infty}\frac{z^\prime(n)}{n}\]
    and 
    \[\Dim_{FS}(S')\geq1-\liminf_{n\to\infty}\frac{z^\prime(n)}{n}\]
\end{corollary}

In \cite{doty2007finite}, it is shown that for any rational number $q$, and real number $x$, the numbers $x, q+x,$ and $qx$ have equal finite-state dimensions and equal finite-state strong dimensions. We now utilize this result to show that linear polynomials with rational coefficients do not change the finite-state dimension nor finite-state strong dimension of Copeland-Erd\H{o}s sequences.

\begin{proposition}
    Let $p(x)=qx+r$ with $q,r\in \Q$, $q\neq 0$. For all $b\geq 2$ and $A\subseteq \N$,
    \[\dim_{FS}(\textup{CE}_b(A))=\dim_{FS}(\textup{CE}_b(p(A)))\] and
    \[\Dim_{FS}(\textup{CE}_b(A))=\Dim_{FS}(\textup{CE}_b(p(A))).\]
\end{proposition}
\begin{proof}
    First note that $\dim_{FS}(\CE_b(A))=\dim_{FS}(\CE_b(A+r))$ for any $r\in\Q$ by Proposition~\ref{density_zero_lemma}. We now consider the case that $q$ is an integer and $r=0$. Let $k\in\Z\setminus\{0\}$ and enumerate $A=\{n_1,n_2,\dots\}$. For each $i\geq 1$, let \[j_i=\big||\sigma_b(n_i)|-|\sigma_b(kn_i)|\big|.\] Let $x\in\R$ be such that \[\sigma_b(x)=0.0^{j_1}\sigma_b(n_1)0^{j_2}\sigma_b(n_2)\dots\] so that $\sigma_b(kx)=0.\CE_b(kA)$ and therefore $\dim_{FS}(\CE_b(kA))=\dim_{FS}(kx)$. By Proposition~\ref{density_zero_lemma}, we have $\dim_{FS}(\CE_b(A))=\dim_{FS}(x)$. If $\dim_{FS}(\CE_b(A))\neq\dim_{FS}(\CE_b(kA))$, then $x$ is a real number such that $\dim_{FS}(x)\neq \dim_{FS}(kx)$ contradicting the main theorem of \cite{doty2007finite}.

    For $q=1/b$, we take $x^\prime\in\R$ to be such that \[\sigma_b(x^\prime)=\sigma_b(n_1^\prime)\sigma_b(n_2^\prime)\dots\] where, for each $i$, $n_i^\prime$ is the nearest multiple of $b$ to $n_i$. By Proposition~\ref{density_zero_lemma}, we have $\dim_{FS}(x)=\dim_{FS}(CE_b(A))$ and $\dim_{FS}(qx)=\dim_{FS}(CE_b(qA))$.

    Exactly similar arguments give the corresponding results for finite-state strong dimension.
\end{proof}

We will now investigate the case for polynomials of degree at least two. To do so we use the following definition due to Besicovitch \cite{besicovitch1935asymptotic}.

\begin{definition}
    A natural number $n$ is $(\epsilon,k)$-normal in base $b\geq 2$ if \[\left|\frac{N(w,\sigma_b(n))}{|\sigma_b(n)|}-\frac{1}{b^k}\right|\leq \epsilon\] for every string $w\in\Sigma_b^k$.
\end{definition}

The following result is well known and appears as Theorem 1 in \cite{pollack2015besicovitch}.

\begin{lemma}\label{constructing-normal}
Consider a sequence $\{a_n\}_{n=1}^\infty$. Suppose that the lengths of the strings $\sigma_b(a_n)$ are growing on average, but that no one length dominates; more precisely, suppose that as $m$ tends to infinity,
    \[m=o\left(\sum_{n=1}^m|\sigma_b(a_n)|\right)\text{ and }m\cdot \max_{1\leq n\leq m}|\sigma_b(a_n)|=O\left(\sum_{n=1}^m|\sigma_b(a_n)|\right).\]
In addition, suppose that for any fixed $\epsilon>0$ and $k\in\N$, almost all $a_n$ are $(\epsilon,k)-$normal, in the sense that the number of $n\leq m$ for which $a_n$ is not $(\epsilon,k)-$normal is $o(m)$ as $m$ tends to infinity, then $\sigma_b(a_1)\sigma_b(a_2)\sigma_b(a_3)\dots$ is normal.
\end{lemma}

The following results also appear in \cite{pollack2015besicovitch} with the first part due to Lemma 4.7 in \cite{bugeaud2012distribution}.

\begin{lemma}\label{frequent-normality}
    Let $\epsilon >0, k\in \mathbb{N}$, and $b\geq 2$ Then.
    \begin{enumerate}
        \item For every $\epsilon>0$ and $k\in \mathbb{N}$ there is a $\delta$ such that the number of integers $n\in[1,m]$ that are not $(\epsilon,k)$-normal is at most $m^{1-\delta}$ for all sufficiently large $m$.
        \item The number of integers $n\in[0,m]$ for which $n^2$ is not $(\epsilon,k)$-normal is $o(m)$ as $m$ goes to infinity.
    \end{enumerate}
\end{lemma}

We now show that polynomials of degree higher than 1 can change the finite-state dimension, even with rational coefficients.
\begin{proposition}\label{poly-changes}
    For every $b\geq 2$ there exists $A\subseteq \N$ and a polynomial $p\in\Q[x]$ of degree 2 so that
        \[0<\dim_{FS}(\CE_b(A))<\dim_{FS}(\CE_b(p(A)))<1\]
\end{proposition}
\begin{proof}
    Note that for any $b\geq 2$, $n\in\N$ and $i\geq |\sigma_b(n)|$ we have \[\sigma_b\left((b^i+n)\right)=10^{i-|\sigma_b(n)|}\sigma_b(n).\]

    Let $p(x)=x^2$. We have \[\sigma_b\left((b^i+n)^2\right)=10^{i-|\sigma_b(2n)|}\sigma_b(2n)0^{i-\sigma_b(n^2)}\sigma_b(n^2).\]
    
    Let $\epsilon_i=2^{-i}$ and $k_i=i$. Then by Lemma \ref{frequent-normality}, for each $i$ there is an $l_i\in \mathbb{N}$ for which there is an $n\in \Sigma_b^l$ with $n,2n$ and $n^2$ all $(\epsilon_i,k_i)$-normal for every $l>l_i$.

    We construct $A$ in stages as follows starting with $i=1$ and $l=l_i$.
    \begin{itemize}
        \item Pick an $n\in \Sigma_b^l$ where $n, 2n$ and $n^2$ are $(\epsilon_i,k_i)$-normal and add $b^l+n$ to $A$. Set $l=l+1$.
        \item Once $l=l_{i+1}$ set $i=i+1$ and go back to the first step.
    \end{itemize}

    Now consider the sequence $S$ formed by stripping the zero block corresponding to each $b^l$ in $\CE_b(A)$. Similarly let $T$ be the sequence obtained by removing the zero block corresponding to the binomial expansion in $\CE_b(p(A))$. By Lemma \ref{constructing-normal}, both $S$ and $T$ are normal. By Proposition \ref{zero-blocks} we have $\dim_{FS}(\CE_b(A))=\Dim_{FS}(\CE_b(A))=0.5$ and $\dim_{FS}(\CE_b(p(A)))=\Dim_{FS}(\CE_b(p(A)))=0.75$
    as $lim_{n\to\infty}\frac{z(n)}{n}$ is $0.5$ for $\CE_b(A)$ and $0.25$ for $\CE_b(p(A))$.
    
\end{proof}

\begin{corollary}
    For all $d\geq 2$, there exists $A\subseteq\N$ and $p\in\Q[x]$ so \[0<\dim_{FS}(\CE_{b}(A))<\dim_{FS}(\CE_{b}(p(A)))<1\]
\end{corollary}
\begin{proof}
    By the binomial theorem, we have for any $i\geq {d\choose \lfloor d/2\rfloor}$, the string $\sigma_b\left((b^i+n)^d\right)$ is given by \[10^{i-\left|\sigma_b\left({d\choose 1}n\right)\right|}\sigma_b\left({d\choose 1}n\right)0^{i-\left|\sigma_b\left({d\choose 2}n^2\right)\right|}\sigma_b\left({d\choose 2}n^2\right)\dots0^{i-\left|\sigma_b\left({d\choose d}n^d\right)\right|}\sigma_b\left({d\choose d}n^d\right).\]

    Theorem 2 of the original work of Davenport and Erd\H{o}s \cite{davenport1952note} states that for any integer-valued polynomial $p(x)$, the number of numbers $n\leq m$ for which $f(n)$ is not $(\epsilon,k)-$normal in base 10 is $O(m)$ as $m\to\infty$. A simple generalization of the construction used in the preceding proposition to the expression given by the binomial theorem and the observation that the result of Davenport and Erd\H{o}s holds in every positive integer base $b\geq 2$ proves the corollary.
\end{proof}

While the above results show rational polynomials can change the dimension, the following shows that this is not always the case.

\begin{proposition}
    For every $s\in [0,1]$, $d\in\mathbb{N}$ and every $b\geq 2$ there is a set $A$ with \[\dim_{FS}(\CE_b(A))=\Dim_{FS}(\CE_b(A))=\dim_{FS}(\CE_b(p(A)))=\Dim_{FS}(\CE_b(p(A)))=s\]
    for a rational polynomial of degree $d$.
\end{proposition}
\begin{proof}
    Set $p(x)=x^d$ and let $(\frac{c_n}{d_n})_{n\in \mathbb{N}}$ where $c_n,d_n\in \mathbb{N}$ be a sequence of rationals converging to $s$ with the property that $d_{n+1}-d_n\leq a$ for some fixed constant $a$ and $lim_{n\to \infty} d_n=lim_{n\to \infty} c_n=\infty$. Note that we do not require $\frac{c_n}{d_n}$ to be written in lowest terms.

    We then construct a set $A$ where $\epsilon_i,k_i$ and $l_i$ are defined in the proof of \ref{poly-changes}. Start with $n=0$
    \begin{enumerate}
        \item Let $i\in \mathbb{N}$ be the largest satisfying $l_i<c_n$.
        \item Pick an $m\in \Sigma_b^{c_n}$ with $m$ and $m^d$ $(\epsilon_i,k_i)$-normal and add $m(b^{(d_n-c_n)})$ to $A$. 
        \item Set $n=n+1$ and go back to step 1.
    \end{enumerate}
    Since $\lim_{n\to\infty}c_n=\infty$ we can remove the $(d_n-c_n)$ 0's at the end of each integer in $\CE_b(A)$ to get a normal sequence by Lemma \ref{constructing-normal} and similarly for removing the $d(d_n-c_n)$ 0's from $\CE_b(p(A))$. Since $d_{n+1}-d_n\leq a$ the limit of $z(n)$ behaves nicely with
    \[\lim_{n\to\infty}\frac{z(n)}{n}=\lim_{n\to\infty}\frac{d_n-c_n}{d_n}\] and moreover 
    \[\lim_{n\to \infty}1-\frac{d_n-c_n}{d_n}=\frac{c_n}{d_n}=s.\]
    Hence, by Proposition \ref{zero-blocks} we have \[\dim_{FS}(\CE_b(A))=\Dim_{FS}(\CE_b(A))=\dim_{FS}(\CE_b(p(A)))=\Dim_{FS}(\CE_b(p(A)))=s.\]
\end{proof}

\begin{remark}
For $0\leq s\leq t\leq 1$ it is also possible to create a set $A$ with \[\dim_{FS}(\CE_b(A))=\dim_{FS}(\CE_b(p(A)))=s\] and \[\Dim_{FS}(\CE_b(A))=\Dim_{FS}(\CE_b(p(A)))=t\] in the same way using one sequence of rationals corresponding to $s$ and another for $t$. Then alternating from padding based on the sequence corresponding to $s$ after the frequency of zeros is arbitrarily close to the correct value to padding according to $t$ and vice versa one can get the $\liminf$ and $\limsup$ to correspond to the correct values.
\end{remark}

The following natural question arises from \cite{szusz1994combinatorial} (Theorem 5): does the normality of $\CE_b(A)$ imply that $\CE_b(p(A))$ is normal? We now answer this question in the negative by demonstrating the existence of a set $A$ such that $A\subseteq\N$ yields a dimension-1 sequence, but $p(A)$ has dimension strictly less than 1. The proof requires a few number-theoretic facts which we present in the appendix.

\begin{lemma}\label{squaredecrease}\footnote{The proof of this lemma is located in the appendix.}
    Let $\epsilon>0, k,b,d\in\mathbb{N}$ with $b,d\geq 2$ and $\gcd(d,b)=1$. Then there exists $N,c\in\mathbb{N}$ and $p\in \mathbb{Q}[x]$ of degree $d$ such that for any $n\geq N$ there exists $m\in \mathbb{N}$  with
    \begin{enumerate}
        \item $|\sigma_b(m)|\geq 2n-\log(n)-c$
        \item $\sigma_b(m)$ $(\epsilon,k)$-normal
        \item  $0^{\frac{n}{2}-c}$ is a substring of $\sigma_b(p(m))$.
    \end{enumerate}
\end{lemma}

\begin{theorem}
    For every $b,d \geq2$ with $\gcd(d,b)=1$, there exists $A\subseteq \mathbb{N}$ $p\in \mathbb{Q}[x]$ of degree $d$ such that 
    \[1=\dim_{FS}(\CE_b(A))>\dim_{FS}(\CE_b(p(A)))\]
    and 
    \[1=\Dim_{FS}(\CE_b(A))>\Dim_{FS}(\CE_b(p(A))).\]
\end{theorem}
\begin{proof}
    It is straightforward to construct a sequence of numbers $\{a_n\}_{n=1}^\infty$ that satisfies Lemma \ref{squaredecrease} for a sequence of $\epsilon_i$ going to 0 and $k_i$ going to infinity in the style of Proposition \ref{poly-changes}. By Lemma \ref{constructing-normal}, the corresponding set $A$ has $1=\dim_{FS}(\CE_b(A))=\Dim_{FS}(\CE_b(A))$. It is also straightforward to see that the length of the numbers in the sequence can be chosen to grow slowly enough that the blocks of approximately $\frac{n}{2}$ zeros in the outputs $p(m)$ of length approximately $2dn$ cause there resulting sequence to consist of $\frac{1}{4d}$ zero blocks as a fraction of the sequence length in the limit. For large enough $l$ this implies that the $l$-block entropy will be less than 1 for sufficiently large $n$ as the probability of the string $0^l$ appearing is larger than $\frac{1}{b^l}$. Thus, 
    \[1>\Dim_{FS}(\CE_b(p(A)))\geq \dim_{FS}(\CE_b(p(A))).\]
\end{proof}

We conclude this section by demonstrating the existence of a sequence of finite-state dimension zero whose finite-state dimension becomes positive under a rational polynomial transformation.

\begin{proposition}
    For all $b\geq 2$, there exists $A\subseteq\N$ and a $p\in\R[x]$ such that $\dim_{FS}(\CE_b(A))=0$ and $\dim_{FS}(\CE_b(p(A))>0$.
\end{proposition}
\begin{proof}
We will employ the use of the polynomial $p(x)=x^2$. First note that \[(1+x+x^2+\dots+x^n)^2=1+2x+3x+\dots+(n+1)x^n+nx^{n+1}+\dots+2x^{2n-1}+x^{2n}\] and thus, if $\sigma_b(n)=(10^i)^j1$ (where exponentiation refers to concatenation), then $\sigma_b(n^2)$ has the form \[10^{\ell_1}\sigma_b(1)0^{\ell_2}\sigma_b(2)\dots0^{\ell_{j+1}}\sigma_b(j+1)0^{\ell_j}\sigma_b(j)\dots0^{\ell_2}\sigma_b(2)0^{\ell_1}\sigma_b(1)\] where $\ell_k=i+1-\sigma_b(k)$.  For example, in base 10, \[1000010000100001^2=1000020000300004000030000200001.\]
Note that if $A=\{n | (10^i)^j1=\sigma_b(n), n\in S\subseteq\N\}$ for some (infinite) $S\subseteq N$, then $\dim_{FS}(CE_b(A))=0$. Let $A$ be the set of numbers $a_k$ such that $\sigma_b(a_i)$ has the form $(10^i)^j1$ and $j$ is such that all numbers with base $b$ expressions of length $i$ appear in $\sigma_b(a_i)$. For each $i$, write $\sigma_b(a_i)=x_i^Ly_ix_i^R$ where $y_i$ is the portion of $\sigma_b(a_i)$ consisting of numbers whose base $b$ expressions have length $i$ and which appear in increasing order in $\sigma_b(a_i)$. We have $|\sigma_b(a_i)|=2i(b^i-2)+i$ and $|x_i^L|=|x_i^R|=i(b^{i-1}-1)$. Calculating the relative fraction of the lengths of strings $\sigma_b(a_i)$ consisting of the substrings $y_i$, we see that \[\lim_{i\to\infty}\frac{(b-1)b^{i-1}i}{2i(b^i-2)+i}=\frac{b-1}{2b}>0.\] By Corollary \ref{garbage-blocks}, the dimension of $\CE_b(p(A))$ is positive.

\end{proof}

\section{Zeta Dimension}

The {\it zeta-dimension} of a set $A\subseteq\N$, \[\Dim_\zeta(A)=\inf\{s|\zeta_A(s)<\infty\}\] where $\zeta_A(s)=\sum_{n\in A}\frac{1}{n^{s}}$, is well-known to have the property that $\Dim_\zeta(A)=1$ implies $CE_b(A)$ is normal. This notion of dimension admits an additional ``entropy characterization" given by \[\Dim_\zeta(A)=\limsup_{n\to\infty}\frac{\log|A\cap\{1,\dots,n\}|}{\log n}\] due to Cahen \cite{cahen1894fonction}. As noted in \cite{gu2005dimensions}, it is natural to define the {\it lower} zeta-dimension, $\dim_\zeta(A)$ by replacing the limit superior with a limit inferior. In \cite{gu2005dimensions}, it is shown that $\dim_{FS}(CE_b(A))\geq\dim_{\zeta}(A)$ and $\Dim_{FS}(CE_b(A))\geq\Dim_\zeta(A)$, and furthermore that for any 
\[
\begin{array}{c c c c c c c}
        &\gamma & \leq & \delta & \leq & 1       \\
        & \rotatebox[origin=c]{90}{\(\leq\)}    &   & \rotatebox[origin=c]{90}{\(\leq\)}              &               \\
       0 \hspace{1em} \leq  & \alpha & \leq & \beta.       &       \\
\end{array}
\]
there exists a set $A\subseteq\N$ with $\dim_\zeta(A)=\alpha$, $\Dim_\zeta(A)=\beta$, $\dim_{FS}(\CE_b(A))=\gamma$ and $\Dim_{FS}(\CE_b(A))=\delta$. The following proposition describes the behavior of both zeta-dimensions under polynomial transformations.

\begin{proposition}
    For all $A\subseteq\N$ and $p\in\R[x]$ with $\deg(p)=d$, \[\dim_\zeta(p(A))=\frac{1}{d}\dim_\zeta(A)\hspace{1em}\text{and}\hspace{1em}\Dim_\zeta(p(A))=\frac{1}{d}\Dim_\zeta(A).\]
\end{proposition}
\begin{proof}
    By the series characterization of zeta dimension, we have \[\Dim_\zeta(p(A))=\inf\{s | \zeta_{p(A)}(s)<\infty\}.\] Noting that, for large $n$, $p(n)\sim Cn^d$ for some $C$, we observe that the above infimum is precisely $d^{-1}\Dim_\zeta(A)$. To show the corresponding result for lower zeta dimension, we utilize the entropy characterization. First note that without loss of generality, $|p(A)\cap\{p(1),\dots,p(n)\}|=|A\cap\{1,\dots,n\}|$ for sufficiently large $n$. Again noting that $p(x)\sim Cn^d$ for large $n$, we have
    \[\dim_\zeta(A)=\liminf_{n\to\infty}\frac{\log(|A\cap\{1,\dots,n\}|)}{\log{n}}=\liminf_{n\to\infty}d\frac{\log(|p(A)\cap\{p(1),\dots,p(n)\}|)}{\log(Cn^d)}.\]
    Now,
    \[\limsup_{n\to\infty}\left|\frac{\log(|p(A)\cap\{p(1),\dots,p(n)\}|)}{\log(Cn^d)}-\frac{\log(|p(A)\cap\{p(1),\dots,p(n+1)-1\}|)}{\log(C(n+1)^d)}\right|=0\] and thus $\dim_\zeta(p(A))=d^{-1}\dim_\zeta(A)$ as desired.
\end{proof}

\section{Future Work}

We view this work as a step toward understanding the way in which the Davenport-Erd\H{o}s theorem is the dimension-1 special case of a more interesting theorem about finite-state dimension. Remaining open questions include: analysis of bounds on $|\dim_{FS}(CE_b(A))-\dim_{FS}(CE_b(p(A)))|$ for polynomials of degree $d\geq 2$, consideration of non-polynomial functions (as in \cite{de2011problem,madritsch2008normality,pollack2015some,szusz1994combinatorial}) in our context, and extension to arbitrary pseudo-polynomials (as in \cite{nakai1990class}).



\printbibliography

@article{champernowne1933construction,
  title={The construction of decimals normal in the scale of ten},
  author={Champernowne, David G},
  journal={Journal of the London Mathematical Society},
  volume={1},
  number={4},
  pages={254--260},
  year={1933},
  publisher={Oxford University Press}
}

@article{besicovitch1935asymptotic,
  title={The asymptotic distribution of the numerals in the decimal representation of the squares of the natural numbers},
  author={Besicovitch, Abram S},
  journal={Mathematische Zeitschrift},
  volume={39},
  pages={146--156},
  year={1935},
  publisher={Springer}
}

@article{copeland1946note,
  title={Note on normal numbers},
  author={Copeland, Arthur H and Erd{\"o}s, Paul},
  journal={Bull. Amer. Math. Soc.},
  volume={52},
  number={12},
  pages={857--860},
  year={1946}
}

@article{davenport1952note,
  title={Note on normal decimals},
  author={Davenport, Harold and Erd{\"o}s, Paul},
  journal={Canadian Journal of Mathematics},
  volume={4},
  pages={58--63},
  year={1952},
  publisher={Cambridge University Press}
}

@article{shiokawa1974remark,
  title={A remark on a theorem of Copeland-Erd{\"o}s},
  author={Shiokawa, Iekata},
  journal={Proceedings of the Japan Academy},
  volume={50},
  number={4},
  pages={273--276},
  year={1974},
  publisher={The Japan Academy}
}

@article{nakai1990class,
  title={A class of normal numbers To the memory of Isamu Kobayashi},
  author={Nakai, Yoshinobu and Shiokawa, Iekata},
  journal={Japanese journal of mathematics. New series},
  volume={16},
  number={1},
  pages={17--29},
  year={1990},
  publisher={The Mathematical Society of Japan}
}

@article{nakai1990classII,
  title={A class of normal numbers, II},
  author={Nakai, YN and Shiokawa, Iekata},
  journal={London Math. Soc. Lecture Note Ser},
  volume={154},
  pages={204--210},
  year={1990}
}

@article{szusz1994combinatorial,
  title={A combinatorial method for constructing normal numbers},
  author={Sz{\"u}sz, Peter and Volkmann, Bodo},
  year={1994},
  publisher={Walter de Gruyter, Berlin/New York Berlin, New York}
}

@inproceedings{gu2005dimensions,
  title={Dimensions of Copeland-Erd{\"o}s sequences},
  author={Gu, Xiaoyang and Lutz, Jack H and Moser, Philippe},
  booktitle={FSTTCS 2005: Foundations of Software Technology and Theoretical Computer Science: 25th International Conference, Hyderabad, India, December 15-18, 2005. Proceedings 25},
  pages={250--260},
  year={2005},
  organization={Springer}
}

@article{madritsch2008normality,
  title={Normality of numbers generated by the values of entire functions},
  author={Madritsch, Manfred G and Thuswaldner, J{\"o}rg M and Tichy, Robert F},
  journal={Journal of number theory},
  volume={128},
  number={5},
  pages={1127--1145},
  year={2008},
  publisher={Elsevier}
}

@article{de2011problem,
  title={On a problem on normal numbers raised by Igor Shparlinski},
  author={De Koninck, Jean-Marie and Katai, Imre},
  journal={Bulletin of the Australian Mathematical Society},
  volume={84},
  number={2},
  pages={337--349},
  year={2011},
  publisher={Cambridge University Press}
}

@article{pollack2015some,
  title={Some normal numbers generated by arithmetic functions},
  author={Pollack, Paul and Vandehey, Joseph},
  journal={Canadian Mathematical Bulletin},
  volume={58},
  number={1},
  pages={160--173},
  year={2015},
  publisher={Cambridge University Press}
}

@article{kozachinskiy2021automatic,
  title={Automatic Kolmogorov complexity, normality, and finite-state dimension revisited},
  author={Kozachinskiy, Alexander and Shen, Alexander},
  journal={Journal of Computer and System Sciences},
  volume={118},
  pages={75--107},
  year={2021},
  publisher={Elsevier}
}

@article{doty2007finite,
  title={Finite-state dimension and real arithmetic},
  author={Doty, David and Lutz, Jack H and Nandakumar, Satyadev},
  journal={Information and Computation},
  volume={205},
  number={11},
  pages={1640--1651},
  year={2007},
  publisher={Elsevier}
}

@article{pollack2015besicovitch,
  title={Besicovitch, Bisection, and the Normality of 0.(1)(4)(9)(16)(25)…},
  author={Pollack, Paul and Vandehey, Joseph},
  journal={The American Mathematical Monthly},
  volume={122},
  number={8},
  pages={757--765},
  year={2015},
  publisher={Taylor \& Francis}
}

@article{bourke2005entropy,
  title={Entropy rates and finite-state dimension},
  author={Bourke, Chris and Hitchcock, John M and Vinodchandran, NV},
  journal={Theoretical Computer Science},
  volume={349},
  number={3},
  pages={392--406},
  year={2005},
  publisher={Elsevier}
}

@article{dai2004finite,
  title={Finite-state dimension},
  author={Dai, Jack J and Lathrop, James I and Lutz, Jack H and Mayordomo, Elvira},
  journal={Theoretical Computer Science},
  volume={310},
  number={1-3},
  pages={1--33},
  year={2004},
  publisher={Elsevier}
}

@article{emile1909probabilites,
  title={Les probabilit{\'e}s d{\'e}nombrables et leurs applications arithm{\'e}tiques},
  author={{\'E}mile Borel, M},
  journal={Rendiconti del Circolo Matematico di Palermo (1884-1940)},
  volume={27},
  number={1},
  pages={247--271},
  year={1909},
  publisher={Springer Milan Milan}
}

@book{wall1949normal,
  title={Normal numbers},
  author={Wall, Donald Dines},
  year={1949},
  publisher={University of California, Berkeley}
}

@book{bugeaud2012distribution,
  title={Distribution modulo one and Diophantine approximation},
  author={Bugeaud, Yann},
  volume={193},
  year={2012},
  publisher={Cambridge University Press}
}

@article{schnorr1972endliche,
  title={Endliche automaten und zufallsfolgen},
  author={Schnorr, Claus-Peter and Stimm, Hermann},
  journal={Acta Informatica},
  volume={1},
  pages={345--359},
  year={1972},
  publisher={Springer}
}

@article{becher2018normal,
  title={Normal numbers and computer science},
  author={Becher, Ver{\'o}nica and Carton, Olivier},
  journal={Sequences, groups, and number theory},
  pages={233--269},
  year={2018},
  publisher={Springer}
}

@article{hitchcock2003fractal,
  title={Fractal dimension and logarithmic loss unpredictability},
  author={Hitchcock, John M},
  journal={Theoretical Computer Science},
  volume={304},
  number={1-3},
  pages={431--441},
  year={2003},
  publisher={Elsevier}
}

@inproceedings{kozachinskiy2019two,
  title={Two characterizations of finite-state dimension},
  author={Kozachinskiy, Alexander and Shen, Alexander},
  booktitle={Fundamentals of Computation Theory: 22nd International Symposium, FCT 2019, Copenhagen, Denmark, August 12-14, 2019, Proceedings 22},
  pages={80--94},
  year={2019},
  organization={Springer}
}

@article{toth2010survey,
  title={A survey of gcd-sum functions},
  author={T{\'o}th, L{\'a}szl{\'o}},
  journal={J. Integer Sequences},
  volume={13},
  number={1},
  year={2010}
}

@inproceedings{cahen1894fonction,
  title={Sur la fonction $\zeta (s)$ de Riemann et sur des fonctions analogues},
  author={Cahen, Eugene},
  booktitle={Annales scientifiques de l'{\'E}cole Normale Sup{\'e}rieure},
  volume={11},
  pages={75--164},
  year={1894}
}

@article{athreya2007effective,
  title={Effective strong dimension in algorithmic information and computational complexity},
  author={Athreya, Krishna B and Hitchcock, John M and Lutz, Jack H and Mayordomo, Elvira},
  journal={SIAM journal on computing},
  volume={37},
  number={3},
  pages={671--705},
  year={2007},
  publisher={SIAM}
}

\appendix

\section{Proof of Lemma \ref{allprefixesproposition}}

\begin{proof}
    Let $l\in\mathbb{N}$. It suffices to show that 
    \[\liminf_{n\to \infty} H_l(S[0...n]) \leq \liminf_{n\to \infty} H_l(T[0...n]).\]
    To see this, let $x=\liminf_{n\to \infty} H_l(S[0...n])$. If $x=0$ then we are done, otherwise let  $x>\epsilon>0$. Then there is $N\in \mathbb{N}$ such that $H_l(S[0...n])>x-\frac{\epsilon}{2}$ for all $n\geq N$. Let $1>p=\frac{x-\epsilon}{x-\frac{\epsilon}{2}}>0$. We now claim that for all $m\in \mathbb{N}$ with $\sum_{n=N}^m1\geq \frac{N^2p}{1-p}$ and all $i\in \mathbb{N}$ with $i\leq m$ that  \[H_l(S[0]S[0...1]...S[0...m]S[0...i]\geq x-\epsilon\] 
    and hence the proposition is true. First note that 
    \[H_l(S[0...N]), ..., H_l(S[0...m])\geq x-\frac{\epsilon}{2}\] 
    and by repeated use of Proposition \ref{concavity} for $H_l$ we get that
    \[H_l(S[0...N]S[0...N+1]....S[0...m])\geq x- \frac{\epsilon}{2}.\]
    One more use of concavity yields
     \begin{align}
        H_l(S[0]...S[0...m]S[0...i])&\geq \frac{0+(H_l(S[0...N]...S[0...m])(|S[0...N]...S[0...m]|)}{|(S[0]...S[0...m]S[0...i]|} \nonumber \\
        &\geq \frac{(x-\frac{\epsilon}{2})(\sum_{n=N}^mn)}{i+\sum_{n=1}^mn} \nonumber \\
        &\geq \frac{(x-\frac{\epsilon}{2})(\frac{N^2p}{1-p})}{N^2+(\frac{N^2p}{1-p})} \nonumber \\
        &=(x-\frac{\epsilon}{2})\frac{(\frac{N^2p}{1-p})}{N^2+(\frac{N^2p}{1-p})}\nonumber \\
        &= (x-\frac{\epsilon}{2})p \nonumber \\
        &= (x-\frac{\epsilon}{2})\frac{(x-\epsilon)}{(x-\frac{\epsilon}{2})} \nonumber \\
        &=x-\epsilon. \nonumber
    \end{align}
\end{proof}

\section{Proof of Lemma \ref{squaredecrease}}

\begin{observation}\label{generators}
    For every base $b\geq 2$ there is a positive constant $c_b$ such that the set $G_n=\{m\in[0,b^n]:\gcd(m,b^n)=1\}$ satisfies 
    \[|G_n|\geq c_bb^n\]
    for all $n\in \mathbb{N}$
\end{observation}
\begin{proof}
    We have 
    \[|G_n|=\varphi(b^n)=b^n\prod_{p|b^n}\left(1-\frac{1}{p}\right)=b^n\prod_{p|b}\left(1-\frac{1}{p}\right)\]
    where $\varphi$ denotes Euler's totient function and $p$ ranges over the primes. Letting 
    \[c_b=\prod_{p|b}\left(1-\frac{1}{p}\right)\]
    the observation holds.
\end{proof}

We will use the following result on the gcd-sum function 
\[P(n)=\sum_{k=1}^n\gcd(k,n)\]
\begin{lemma}\cite{toth2010survey}\label{gcdsum}
    Let $\tau(n)$ be the number of divisors of $n$. Then
    \[P(n)\leq n\tau(n)\]
    for all $n\in\mathbb{N}$.
\end{lemma}

\begin{proposition}\label{collisions}
    If a set $S=\{v_1,\dots v_n\}$ of $n$ values has $C$ pairs $\{v_i,v_j\}$ where $v_i=v_j$, then there are at least 
    \[\frac{n^2}{2C+n}\]
    unique values in $S$.
\end{proposition}
\begin{proof}
    Let $k$ be the number of unique values and $S_1,\dots S_k$ be the subsets of $S$ where every element in the subset has the same value. It is routine to verify that the minimum value of $k$ occurs when the number of collisions and hence the size of each $S_i$ is the same.

    Therefore to obtain a lower bound assume that there are $k$ sets each with $\frac{n}{k}$ elements and $\frac{C}{k}$ collisions. For this to be true we need 
    \[{\frac{n}{k} \choose 2}=\frac{C}{k}\]
    and solving for $k$ yields 
    $k=\frac{n^2}{2C+n}$.
\end{proof}
We now prove Lemma \ref{squaredecrease}.
\begin{proof}
    We will consider numbers $m=b^nx+y$ divided in two parts so that $\sigma_b(m)=\sigma_b(x)\sigma_b(y)$ and use the polynomial $p(m)=m^d$. By expanding the formula for $m^d=(b^nx+y)^d$ we see that the $n$ least significant digits depend only on $y^d$. The next $n$ significant digits are the result of the addition of a middle portion of $y^d$ and $dxy^{d-1}$. We will show that for a sufficiently large portion of possible $y's$, there is a large portion of $x's$ so that both $x$ and $y$ are $(\epsilon,k)$-normal and the left of this block is all $0$'s. 

    We will consider $y$ from the set $G_n$ defined in Observation \ref{generators}. Let $\hat{y}$ be the $n$ bits of $y^d$ so that $\sigma(\hat{y}+dxy^{d-1}$ mod $b^n)$ appears in $\sigma_b(b^nx+y)^d$ ignoring the potential carry. Then to have the left half to the digits of $p(m)$ inside this block be 0 excluding a constant amount that may be altered by the carry, we need 
    \[\hat{y}+dxy^{d-1}\equiv t \mod b^n \]
    for some $t\in[0,b^{n/2})$. Note that $G_n$ consists of the multiplicative group of $(\mathbb{Z}/b^n\mathbb{Z})^\times$ so $(y^{d-1})^{-1}$ is well defined and also in $G_n$. Therefore we get
    \[dx_{y,t}\coloneq dx\equiv -\hat{y}(y^{d-1})^{-1}+t(y^{d-1})^{-1}\mod b^n. \] 
    Note that since $\gcd(d,b)=1$ there is a unique solution $x\in[0,b^{n/2})$. 
    Let $c_y\coloneq -\hat{y}(y^{d-1})^{-1}d^{-1}$ and $g_y\coloneq (y^{d-1})^{-1}d^{-1}$ so that $x_{y,t}=c_y+tg_y \mod b^n$ and let  $x_{y,t}=l_{y,t}b^{n/2}+r_{y,t}$ be the splitting so that $l_{y,t},r_{y,t}\in [0,b^{n/2})$. Define $L_y=\{l_{y,t}:t\in[0,b^{n/2})\}$ and $R_y=\{r_{y,t}:t\in[0,b^{n/2})\}$.

    \textbf{Claim:}
        There are constants $c,a>0$ such that at least $\frac{|G_n|}{2}$ of the values $y\in G_n$ have $|L_y|\geq \frac{cb^{n/2}}{n^a}$ for sufficiently large $n$.

    Suppose the claim is true. Then note that $r_{y,t}=c_y+tg_y \mod b^{n/2}$ and moreover that since $y\in G_n$ so is $(y^{d-1})$ along with its inverse $(y^{d-1})^{-1}$ and $g_y$. Thus, we also have $\gcd(g_y,b^{n/2})=1$ so as $t$ ranges over $[0,b^{n/2})$ we get all possible values in $R_y$ and in particular no two values of $t$ have the same corresponding $r_{y,t}$. 

    By Lemma \ref{frequent-normality} we can choose an $(\epsilon,k)$-normal $y$ satisfying the above claim for sufficiently large $n$. For that $y$, $L_y$ contains over twice as many numbers as the amount of non $(\epsilon,k)$-normal numbers, so after removing the bad ones we still have more than the amount of bad ones meaning one of those $l_{y,t}$ coincide with an $(\epsilon,k)$-normal $r_{y,t}$.

    Concatenating these three $(\epsilon,k)$-normal strings yields an $(\epsilon,k)$-normal $m$ of length $2n$ with potentially leading zeros. 

    Finally note that there are $\Omega(\frac{b^{n/2}}{n^a})$ acceptable choices for the element in $L_y$ Therefore we need approximately $\log_b(\frac{b^{n/2}}{n^a})=\frac{n}{2}-\frac{a}{\log(b)}\log(n)$ bits to represent the largest which can then have at most $O(\log(n))$ leading zeros.

    \textbf{Proof of Claim:} 
        We start by counting the number of total collision triples $\{y,t_1,t_2\}$  where $l_{y,t_1}=l_{y,t_2}$ with $t_1>t_2$. Note that if $l_{y,t_1}=l_{y,t_2}$ then $x_{y,t_1}-x_{y,t_2}=r_{y,t_1}-r_{y,t_2}$. Moreover, $x_{y,t_1}-x_{y,t_1}= (t_1-t_2)g_y \mod b^n$. Thus, letting $s=t_1-t_2$ and $r=r_{y,t_1}-r_{y,t_2}$, if there is a collision then 
        \[sg_y \equiv r \mod b^{n}\]
        for some $r$ with $|r|<b^{n/2}$.
        Recall that this has either 0 or $\gcd(s,b^n)$ solutions. Thus, if we fix both $r$ and $s$ there are at most $\gcd(s,b^n)$ solutions $g_y$.
        As there are less than $2b^{n/2}$ choices for $r$ we get a maximum of $2b^{n/2}\gcd(s,b^n)$ solutions over all possible $r$ and fixed $s$.

        Now note that there are also at most $b^{n/2}$ possibilites of for $t_1$ and $t_2$ with $t_1-t_2=s$. Thus, the total number of collisions is at most  
        \[C=\sum_{s=1}^{b^{n/2}}2b^n\gcd(s,b^n).\]
        By Lemma \ref{gcdsum} we have 
        \[C\leq 2b^{3n/2}\tau(b^{n/2}).\]

        By Observation \ref{generators}, there are at least $c_bb^n$ values of $g_y$ since if $y_1\neq y_2$ then $g_{y_q}\neq g_{y_2}$, so the average amount of collisions for a $g_y$ is at most
        \[\frac{2b^{3n/2}\tau(b^{n/2})}{c_bb^n}=\frac{2b^{n/2}\tau(b^{n/2})}{c_b}.\]
        Thus, note that for at least half of the $g_y$, the amount of collisions is at most 
        \[C'=\frac{4b^{n/2}\tau(b^{n/2})}{c_b}.\]

        By Proposition  \ref{collisions}, for each of these $y$ we have
        \[|L_y|\geq \frac{b^n}{2\frac{4b^{n/2}\tau(b^{n/2})}{c_b}+b^{n/2}}=\frac{b^{n/2}}{\frac{8\tau(b^{n/2})}{c_b}+1}\geq \frac{c_bb^{n/2}}{9\tau(b^{n/2})}\]

        Lastly note that that if $b$ has prime factorization $b=p_1^{c_1}\dots p_k^{c_k}$ then $\tau(b^n)=(c_1n+1)\cdot(c_2n+1)\dots(c_kn+1)$ and in particular $9\tau(b^{n/2})$ is bounded by $n^a$ for some fixed $a$ and sufficiently large $n$.
\end{proof}

\end{document}